\documentclass[utf8]{amsart}
\usepackage{a4wide}
\usepackage{amsmath}
\usepackage{amsthm}
\usepackage{amssymb}
\usepackage{braket}
\usepackage{commath}
\usepackage{mathrsfs}
\usepackage{mathabx}
\usepackage[onehalfspacing]{setspace}
\usepackage{url,hyperref,microtype,subcaption}
\usepackage{bbm}
\usepackage{mathtools}
\usepackage{nicefrac}

\DeclareMathOperator{\spread}{spread}

\usepackage{graphicx}
\usepackage[font=small,skip=0pt]{caption}

\theoremstyle{plain}
\newtheorem{theorem}{Theorem}
\newtheorem*{theorem*}{Theorem}
\newtheorem{proposition}[theorem]{Proposition}
\newtheorem*{proposition*}{Proposition}

\newtheorem{corollary}[theorem]{Corollary}
\newtheorem*{corollary*}{Corollary}

\newtheorem{lemma}[theorem]{Lemma}
\newtheorem*{lemma*}{Lemma}

\newtheorem*{observation*}{Observation}

\newtheorem*{conjecture*}{Conjecture}

\newtheorem*{question*}{Question}

\newtheorem*{questions*}{Questions}

\newtheorem*{problem*}{Problem}

\newtheorem*{problems*}{Problems}

\newtheorem*{openproblem*}{Open Problem}

\newtheorem{definition}[theorem]{Definition}
\newtheorem*{definition*}{Definition}

\newtheorem*{example*}{Example}

\newtheorem*{exercise*}{Exercise}

\newtheorem{remark}[theorem]{Remark}
\newtheorem*{remark*}{Remark}

\newtheorem*{remarks*}{Remarks}

\newtheorem*{claim*}{Claim}

\newcommand{\subclass}[1]{}

\newcommand{\enumTi}[1]{\renewcommand{\theenumi}{#1}}

\newcommand{\alphenumi}{\enumTi{\alph{enumi}}}

\newcommand{\romenumi}{\enumTi{\roman{enumi}}}

\newlength{\hspaceforlengthglumpf}

\newcommand{\comment}[1]{\text{\footnotesize[#1]}}

\DeclareMathOperator{\id}{id}

\DeclareMathOperator{\tr}{tr}

\newcommand{\lt}{\left}
\newcommand{\rt}{\right}

\newcommand{\nfrac}[2]{{\nicefrac{#1}{#2}}}

\newcommand{\CC}{\mathbb{C}}

\newcommand{\RR}{\mathbb{R}}

\newcommand{\ZZ}{\mathbb{Z}}

\newcommand{\eps}{\varepsilon}

\newcommand{\ketbra}[2]{\lvert{#1}\rangle\langle{#2}\rvert}

\newlength{\algotabbingwidth}
\renewcommand{\emph}[1]{\textsl{#1}}

\def\firstAuthorLast{Sample {et~al.}} %use et al only if is more than 1 author
\def\Authors{Francisco Javier Gil Vidal\,$^{1}$
 Dirk Oliver Jim Theis\,$^{1,2,*}$}

\begin{document}
\onecolumn
%\firstpage{1}

\title{Input Redundancy for Parameterized Quantum Circuits}

\author[\firstAuthorLast ]{\Authors} %This field will be automatically populated
\address{} %This field will be automatically populated
\maketitle

\begin{abstract}
  One proposal to utilize near-term quantum computers for machine learning are Parameterized
  Quantum Circuits (PQCs).  There, input is encoded in a quantum state, parameter-dependent unitary
  evolution is applied, and ultimately an observable is measured.  In a hybrid-variational
  fashion, the parameters are trained so that the function assigning inputs to expectation values
  matches a target function.

  The \textit{no-cloning principle} of quantum mechanics suggests that there is an advantage in
  redundantly encoding the input several times.  In this paper, we prove lower bounds on the
  number of redundant copies that are necessary for the expectation value function of a
  PQC to match a given target function.

  We draw conclusions for the architecture design of PQCs.

  \par\noindent
  \textbf{Keywords:} Parameterized Quantum Circuits, Quantum Neural Networks, near-term
  quantum computing; lower bounds.
\end{abstract}

\section{Introduction}
Quantum Information Processing proposes to exploit quantum physical phenomena for the purpose of
data processing.   Conceived in the early 80's \cite{feynman1982simulating, manin1980vychislimoe},
recent breakthroughs in building controllable quantum mechanical systems have led to an explosion
of activity in the field.

Building quantum computers is a formidable challenge --- but so is designing algorithms which,
when implemented on them, are able to exploit the advantage that quantum computing is widely
believed by experts to have over classical computing on some computational tasks.   A particularly
compelling endeavor is to make use of \textit{near-term quantum computers}, which suffer from
limited size and the presence of debilitating levels of quantum noise.   The field of algorithm
design for Noisy Intermediate-Scale Quantum (NISQ) computers has scrambled over the last few years
to identify fields of computing, paradigms of employing quantum information processing, and
commercial use-cases in order to profit from recent progress in building programmable quantum
mechanical devices --- limited as they may be at present~\cite{mohseni2017commercialize}.

One use-case area where quantum advantage might materialize in the near term is that of Artificial
Intelligence~\cite{mohseni2017commercialize,perdomo2018opportunities}.  The hope is best reasoned
for generative tasks: several families of probability distributions have been theoretically proven
to admit quantum algorithms for efficiently sampling from them, while no classical algorithm is
able or is known to be able to perform that sampling task.  Boson sampling is probably the most
widely known of these sampling tasks, even though the advantage does not seem to persist in the
presence of noise (cf.~\cite{neville2017classical}); examples of some other sampling procedures
can be found in references \cite{bremner2016average} and~\cite{farhi2016quantum}.

Promising developments have also been made available in the case of quantum circuits that can be
iteratively altered by manipulation of one or several parameters: Du et
al.~\cite{du2018expressive} consider so-called \textit{Parameterized Quantum Circuits (PQCs)} and
find that they, too, yield a theoretical advantage for generative tasks.  PQCs are occasionally
referred to as \textit{Quantum Neural Networks (QNNs)} (e.g., in~\cite{farhi1802classification})
when aspects of non-linearity are emphasized, or as \textit{Variational Quantum
  Circuits}~\cite{mcclean2016theory}.  We stick to the term PQC in this paper, without having in
mind excluding QNNs or VQCs.

The PQC architectures which have been considered share some common characteristics, but an
important design question is how the input data is presented.  Input data refers either to a
feature vector, or to output of another layer of a larger, potentially hybrid quantum-classical
neural network.  The fundamental choice is whether to encode digitally or in the amplitudes of a
quantum state.  Digital encoding usually entails preparing a quantum register in states
$\ket{b^x}$, where $b^x\in\{0,1\}^n$ is binary encoding of input datum~$x$.  Encoding in the
amplitudes of a quantum state, on the other hand, refers to preparing a an $n$-bit quantum
register in a state of the form $\ket{\phi^x} := \sum_{j=0}^{2^n-1} \phi_j(x) \ket{j}$, where
$\phi_j$, $j=0,\dots,2^n-1$ is a family of encoding functions which must ensure that
$\ket{\phi^x}$ is a quantum state for each~$x$, i.e., that $\sum_j \abs{\phi_j(x)}^2=1$ holds for
all~$x$.  We refer the reader to the discussion of these concepts in~\cite{schuld2018supervised}
for further details.

The present paper deals with redundancy in the input data, i.e., giving the same datum several
times.   The most straightforward concept here is that of ``tensorial''
encoding~\cite{Schuld-Bocharov-Svore-Wiebe:circcentric:2018}.   Here, several quantum registers are
prepared in a state which is the tensor product of the corresponding number of \emph{identical}
copies of a data-encoding state, i.e., $\ket{\phi^x}\otimes\dots\otimes\ket{\phi^x}$.
For example, Mitarai et al.~\cite{Mitarai-Negoro-Kitagawa-Fujii:q-circ-learn:2018}, propose the
following construction: To encode a real number $x$ close to~$0$, they choose the state
\begin{equation}\label{eq:minkif-input}
  \ket{\phi^x}
  = R_y(\arcsin(x)/2) \ket{0}
  =
  \sin(\arcsin(x)/2) \ket{0} + \cos(\arcsin(x)/2) \ket{1},
\end{equation}
where $R_y(\theta) := e^{-i\theta \sigma_Y/2}$ is the 1-qubit Pauli rotation around the $Y$-axis
(and $\sigma_Y$ the Pauli matrix).   But then, to construct a PQC that is able to learn polynomials
of degree~$n$ in a single variable, they encode the polynomial variable~$x$ into~$n$ identical
copies, $\bigotimes_{j=1}^n \ket{\phi^x}$.   It is noteworthy, and the starting point of our
research, that the number of times that the input, $x$, is encoded redundantly, depends on the
complexity of the learning task.

Encoding the input several times redundantly, as in tensorial encoding, is probably motivated by
the quantum no-cloning principle.   While classical circuits and classical neural networks can have
\textit{fan-out} --- the output of one processing node (gate, neuron, \dots) can be the input to
several others --- the no-cloning principle of quantum mechanics forbids to duplicate data which
is encoded in the amplitudes of a quantum state.   This applies to PQCs, and, specifically, to the
input that is fed into a PQC, if the input is encoded in the amplitudes of input states.

\subsection{The research presented in this paper.}
The no-cloning principle suggests that duplicating input data redundantly is unavoidable.  The
research presented in this paper aims to lower bound \emph{how often} the data has to be
redundantly encoded, if a given function is to be learned.  The novelty in this paper lies in
establishing that these lower bounds are possible.  For that purpose, the cases for which we prove
lower bounds are natural, but not overly complex, thus highlighting the principle over the
application.

The objects of study of this paper are PQCs of the following form.  The input consists of a single
real number~$x$, which is encoded into amplitudes by applying a multi-qubit Hamiltonian evolution
of the form $e^{-i\eta(x) H}$ at one point (\textit{no redundancy}), or several points in the
quantum circuit.  The function~$\eta$ and Hamiltonian~$H$ may be different at the different points
the quantum circuit.

Hence, our definition of ``input'' is quite general, and allows, for example, that the input is
given in the middle of a quantum circuit --- mimicking the way how algorithms for fault-tolerant
quantum computing operate on continuous data: the subroutine for accessing the data be called
repeatedly; cf., e.g., the description of the input oracles in \cite{van2018improvements}.  It
should be pointed out, however, that general state preparation procedures as in
\cite{harrow-hassidim-lloyd:hhl:2009} and~\cite{Schuld-Bocharov-Svore-Wiebe:circcentric:2018}
cannot not be studied with the tools of this paper, because they apply many operations with
parameters derived from a \emph{collection} of inputs, instead of a \emph{single} input.

Our lower bound technique is based on Fourier analysis.

\subsection{Example}
Take, as example, the parameterized quantum circuits of Mitarai et
al.~\cite{Mitarai-Negoro-Kitagawa-Fujii:q-circ-learn:2018} mentioned above.  Comparing
with~\eqref{eq:minkif-input} shows: The single real input~$x$ close to~$0$ is is prepared by
performing, at $n$~different positions in the quantum circuit, Hamiltonian evolution
$e^{-i\eta_j(x) H_j}$ with $\eta_j(x) := \arcsin(x)$, and $H_j := \sigma_Y/2$, for $j=1,\dots,n$.

We say that the input~$x$ to the quantum circuits of Mitarai et al.\ are encoded with
\textit{input redundancy~$n$} --- meaning, the input is given~$n$ times.

The example highlights the ostensible wastefulness of giving the same data~$n$ times, and the
question naturally arises whether a more clever application of possibly different rotations would
have reduced the amount of input redundancy.

In the case of Mitarai et al.'s example, it can easily be seen --- from algebraic arguments
involving the quantum operations which are performed --- that, in order to produce a polynomial of
degree~$n$, redundancy~$n$ is best possible for the particular way of encoding the value~$x$
\textsl{by applying the Pauli rotation to distinct qubits}, we leave that to the reader.  However,
already the question whether by re-using the same qubit a less ``wasteful'' encoding could have
been achieved is quite not so easy.  Our Fourier analysis based techniques give lower bounds for
more general encodings, in particular, for applying arbitrary single-qubit Pauli rotations to an
arbitrary set of qubits at arbitrary time during the quantum circuit.

Fig.~\ref{fig:qnn} next page shows the schematic of quantum circuits with input~$x$.  The setup
resembles that of a neural network layer.  The $j$'th ``copy'' of the input is made available in
the quantum circuit by, at some time, performing the unitary operation $e^{2\pi i \eta_j(x) H_j}$
on one qubit, where $\eta_j(x) = \varphi(a_j x + b_j)$, for an ``activation function'' $\varphi$.
(We switch here to adding the factor $2\pi$, to be compatible with our Fourier approach.)

In the above-mentioned example in~\cite{Mitarai-Negoro-Kitagawa-Fujii:q-circ-learn:2018}, the
activation function is $\varphi := \arcsin$.  Fig.~\ref{fig:qnn} aims at making clear that the
input can be encoded by applying different unitary operations to different qubits, or to the same
qubit several times, or any combination of these possibilities.  Generalization of our results to
several inputs is straightforward, if the activation functions in Fig.~\ref{fig:qnn} have a single
input.

\begin{figure}[ht]
  \centering
  \hspace*{-10mm} \includegraphics[scale=1.15]{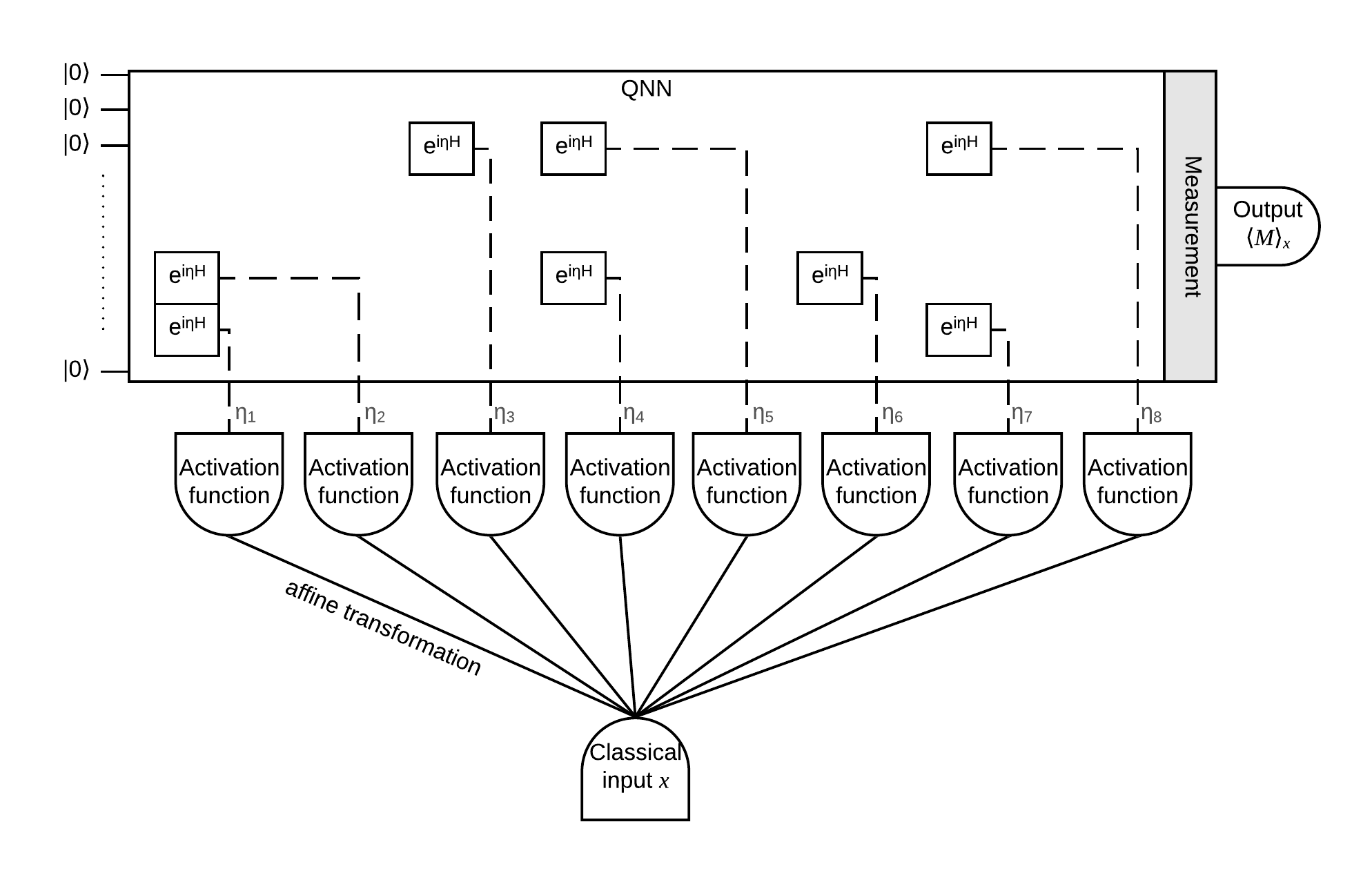}
  \caption{Schematic for the PQCs we consider.  The classical input~$x$, after being subjected to
    a transformations $\eta_j(x) = \varphi(a_j x + b_j)$ with ``activation function'' $\varphi$,
    is fed into the QNN/PQC through Hamiltonian evolution operations $e^{2\pi i \eta_j(x) H_j}$.
    Quantum operations which do not participate in entering the input data into the quantum
    circuit are not shown; these include the operations which depend on the \textit{training
      parameters} $\theta$.}
  \label{fig:qnn}
\end{figure}

\subsection{Our results}
As hinted above, our intention with this paper is to establish, in two natural examples, the
possibility of proving lower bounds on input redundancy.  The first example is what we call
``linear'' input encoding, where the activation function is $\varphi(x) = x$.  The second example
is~\cite{Mitarai-Negoro-Kitagawa-Fujii:q-circ-learn:2018} approach, where the activation function
is $\varphi(x) = \arcsin(x)$.

For both examples, we prove lower bounds on the input redundancy in terms of linear-algebraic
complexity measures of the target function.  We find the lower bounds to be logarithmic, and the
bounds are tight.

To the best of our knowledge, our results give the first quantitative lower bounds on input
redundancy.  These lower bounds, as well as other conclusions derived from our constructions,
should directly influence design decisions for quantum neural network architectures.

\subsection{Paper organization}
In the next section we review the background on the PQC model underlying our results.  Sections
\ref{sec:linear} and~\ref{sec:arcsin} contain the results on linear and arcsine input encoding,
respectively.  We close with a discussion and directions of future work.

\section{Background}
\subsection{MiNKiF PQCs}
We now describe parameterized quantum circuits (PQCs) in more detail.
Denote by
\begin{equation}\label{eq:e^iH}
  U_H(\alpha)\colon \rho \mapsto e^{-2\pi i \alpha H} \rho e^{2\pi i \alpha H}
\end{equation}
the quantum operations of an evolution with Hamiltonian~$H$ (operating on some set of qubits); the
$2\pi$ factor is just a convenience for us and introduces no loss of generality.  Following, in
spirit, \cite{Mitarai-Negoro-Kitagawa-Fujii:q-circ-learn:2018}, in this paper we consider quantum
circuits which apply quantum operations each of which is one of the following:
\begin{enumerate}
\item\label{enum:qopt-type:nopar:input}%
  An operation as in~\eqref{eq:e^iH}, with a parameter $\alpha := \eta$ which will encode input,
  $x$ (i.e., $\eta$ is determined by~$x$);
\item\label{enum:qopt-type:nopar:train}%
  An operation as in~\eqref{eq:e^iH}, with a parameter $\alpha := \theta$ which will be
  ``trained'' (we refer to these parameters as the \textit{training parameters});
\item\label{enum:qopt-type:nopar}%
  Any quantum operation not defined by any parameter (although its effect can \textsl{depend} on
  $\theta$, $\eta$, e.g., via dependency on measurement results).
\end{enumerate}

Denote the concatenated quantum operation by $\mathcal{E}(\eta,\theta)$.
Now let~$M$ be an observable, and consider its expectation value on the state which results if the
parameterized quantum circuit is applied to a fixed input state $\rho_0$, e.g.,
$\rho_0 := \ketbra{0}{0}$.  We denote the expectation value with parameters set to $\eta,\theta$ by~$f(\eta,\theta)$:
\begin{equation}\label{eq:qcirc-fn}
  f\colon
  \RR^n\times\RR^m \to \RR\colon
  (\eta,\theta) \mapsto
  = \tr( M  \mathcal{E}(\eta,\theta) \rho_0 ).
\end{equation}

The PQCs could have multiple outputs, but we do not consider that in this paper.  We refer to PQCs
of this type as \textit{MiNKiF} PQCs, as~\cite{Mitarai-Negoro-Kitagawa-Fujii:q-circ-learn:2018}
realized the fundamental property
\begin{equation}
  \partial_{\theta_j} f(\eta,\theta)
  = \pi \bigl(  f(\eta,\theta+\tfrac14 e_j) - f(\eta,\theta-\tfrac14 e_j) \bigr),
\end{equation}
where $e_j$ is the vector with a~$1$ in position~$j$ and $0$ otherwise.  This equation
characterizes trigonometric functions.  (The same relation holds obviously for derivatives in
$\eta_j$ direction.)

The setting we consider in this paper is the following.
\begin{itemize}
\item The \textbf{training parameters, $\theta$, have been trained perfectly and are thus ignored}, in other
  words, omitting the $\theta$ argument, we conveniently consider~$f$ to be a function defined
  on~$\RR^n$ (instead of on $\RR^n\times\RR^m$);
\item The inputs, $x$, are real numbers;
\item The parameters $\eta$ of~$f$ are determined by~$x$, i.e., $\eta$ is replaced by
  $(\varphi(a_1 x + b_1),\dots,\varphi(a_n x + b_n))$, where $a,b\in\RR^n$; in other words, we
  study the function
  \begin{equation*}
    \RR\to\RR\colon x \mapsto f(\varphi(a_1 x + b_1),\dots,\varphi(a_n x + b_n)).
  \end{equation*}
\end{itemize}
We allow $a,b$ to depend on the target function.\footnote{Cf.~Remarks \ref{rem:Var-Inp-Enc:lin}
  and~\ref{rem:Var-Inp-Enc:arcsin}.  Indeed, our analysis suggests that the $a,b$ should be
  training parameters if the goal is to achieve high expressivity; see the Conclusions.}

This setting is restrictive only in as far as the input is one-dimensional; the reason for this
restriction is that this paper aims to introduce and demonstrate a concept, and not be
encyclopedic or obtain the best possible results.

This setting clearly includes the versions of amplitude encoding discussed in the introduction by
applying operations $U_{H_j}(\varphi(a_j x + b_j))$ to $\ketbra00$ states (of appropriately many
qubits) for several $j$'s, with suitable~$H_j$'s.  However, the setting is more general in that it
doesn't restrict to encode the input near the beginning of a quantum circuit, indeed, the order of
the types of quantum operations is completely free.

To summarize, we study the functions
\begin{subequations}\label{eq:qopt-E}
  \begin{equation}\label{eq:qopt-E:fn}
    f\colon \RR\to\RR\colon
    \eta \mapsto
    f(\eta)
    :=
    \tr \Bigl( M\, V_n \, U_n(\eta_n) \, V_{n-1} \, U_{n-1}(\eta_{n-1}) \; \dots \; V_1 \, U_1(\eta_1) \, V_0 \, \rho_0 \Bigr)
  \end{equation}
  where
  \begin{equation}\label{eq:qopt-E:Ham}
    U_1:=U_{H_1}, \dots, U_n := U_{H_n}
    \\
    \text{for Hamiltonians } H_j, \text{ $j=1,\dots,n$.}
  \end{equation}
\end{subequations}
and
\begin{equation}\label{eq:qopt-E:fn-x}
  \RR\to\RR \colon x \mapsto f(\eta(x))
\end{equation}
where $\eta\colon\RR\to\RR^n\colon x\mapsto \varphi(a_1x+b_1),\dots,\varphi(a_nx+b_n)$.
Then we ask the question: How large is the space of the $x\colon f(\eta(x))$, for a fixed
activation function $\varphi$, but variable vectors $a,b\in\RR^n$?

\subsection{Fourier calculus on MiNKiF circuits}
This paper builds on the simple observation of~\cite{GilVidal-Theis:CalcPQC:2018} that, under
assumptions which are reasonable for near-term gate-based quantum computers, the Fourier spectrum,
in the sense of the Fourier transform of tempered distributions, is finite and can be understood
from the eigenvalues of the Hamiltonians.  In particular, if, for each of the Hamiltonians~$H_j$,
$j=1,\dots,n$, the differences of the eigenvalue of~$H_j$ are integer multiples of a positive
number $\kappa_j$, then $\eta\mapsto f(\eta)$ is periodic.

Take, for example, the case of Pauli rotations ($e^{-2\pi i \sigma_\star/2}$ in our notation):
There, each of the $H_j$ is of the form $\sigma_{u_j}/2$ (with $u_j\in\{x,y,z\}$).  The
eigenvalues of $H_j$ are $\pm\nfrac12$, the eigenvalue differences are $0,\pm 1$, and
$f\colon\RR^n\to\RR$ is 1-periodic\footnote{This is where the factor $2\pi$ in the exponent is
  used.} in every parameter, with Fourier spectrum contained in
\begin{equation}\label{eq:Fou-spec-simple}
  \ZZ_3^n := \{0,\pm1\}^n.
\end{equation}

More generally, if the $H_j$ have eigenvalues, say, $\lambda_j^{(0)}\in\RR$ and
$\lambda_j^{(s)} = \lambda_j^{(0)}+s$ for $s=1,\dots,K_j$, then the eigenvalue differences are
$\{-K_j,\dots,K_j\}$, and $f\colon\RR^n\to\RR$ is 1-periodic in every parameter, with Fourier
spectrum contained in $\prod_{j=1}^n \{-K_j,\dots,K_j\}$.

We refer to~\cite{GilVidal-Theis:CalcPQC:2018} for the (easy) details.  In this paper, focusing on
the goal of demonstrating the possibility to prove lower bounds on the input redundancy, we mostly
restrict to 2-level Hamiltonians with eigenvalue difference~1 (such as one-half times a tensor
product of Pauli matrices), which gives us the nice Fourier spectrum~\eqref{eq:Fou-spec-simple},
commenting on other spectra only \textsl{en passant}.

For easy reference, we summarize the discrete Fourier analysis properties of the expectation value
functions that we consider in the following remark.  The proof of the equivalence of the three
conditions is contained in the above discussions, except for the existence of a quantum circuit
for a given multi-linear trigonometric polynomial, for which we defer~\cite{GilVidal:PhD:2020}, as
it is not the topic of this paper.

\begin{remark}\label{rem:fourier-expansion}
  The following three statements are equivalent.  If they hold, we refer to the function as an
  \textit{expectation value function,} for brevity (suppressing the condition on the eigenvalues
  of the Hamiltonians).  The input redundancy of the function is~$n$.

  \begin{enumerate}
  \item\label{rem:fourier-expansion:expval}%
    The function~$f$ is of the form~\ref{eq:qopt-E}, where the $H_j$, $j=1,\dots,n$, have
    eigenvalues $\pm\nfrac12$.

  \item\label{rem:fourier-expansion:C}%
    The function~$f$ is a real-valued function $\RR^n\to\RR$ which is $1$-periodic in every
    parameter, and its Fourier spectrum is contained in $\ZZ_3^n$.  Hence,
    \begin{equation}\label{eq:f-FouExp-C}
      f(\eta)
      =
      \sum_{w \in \ZZ_3^n} \hat f(w) e^{2\pi i w \bullet \eta}
    \end{equation}
    where $w \bullet \eta := \sum_{j=1}^n \eta_j w_j$ is the dot product (computed in $\RR$), and
    $\hat f$ the usual periodic Fourier transform of~$f$, i.e.,
    $\hat f(w) = \int_{[0,1]^n} e^{-2\pi i w \bullet \eta}f(\eta)\,d\eta$.

  \item\label{rem:fourier-expansion:R}%
    The function~$f$ is a multi-linear polynomial in the sine and cosine functions, i.e.,
    \begin{equation}\label{eq:f-FouExp-R}
      f(\eta)
      =
      \sum_{\tau\in \{1,\cos,\sin\}^n} \tilde f_{\tau} \, \prod_{j=1}^n \tau_j(2\pi \eta_j),
    \end{equation}
    (where ``$1$'' under the sum denotes the all-1 function).
\end{enumerate}
\end{remark}

\section{Linear input encoding}\label{sec:linear}
We start discussing the case where the input parameters are affine functions of the input
variable, e.g., $\varphi=\id$ and $\eta(x) = x\cdot a + b$ for some $a,b\in\RR^n$, so the input
redundancy is~$n$.

For $a\in\RR^n$ define
\begin{subequations}\label{eq:def-spread}
  \begin{align}
    K_a        &:= \bigl\{ w\bullet a \bigm| w\in\ZZ_3^n \bigr\}, \\
    \spread(a) &:= \frac{1}{2} \abs{ K_a\setminus\{0\} },
  \end{align}
\end{subequations}
with $\abs{\cdot}$ denoting set cardinality; we refer to $\spread(a)$ as the \textit{spread} of $a$.
We point the reader to the fact that $K_a$ is symmetric around $0\in\RR$ and $0\in K_a$, so that
the spread is a nonnegative integer.

For every $k\in\RR$, consider the function
\begin{equation}\label{eq:def:chi_k}
  \chi_k\colon  \RR \to \CC \colon  t \mapsto e^{2\pi i k t}.
\end{equation}
\newcommand{\ContFun}{\mathscr C}%
These functions are elements of the vector space $\CC^\RR$ of all complex-valued functions on the
real line.  We note the following well-known fact.

\begin{lemma}\label{lem:chindependence}
  The functions $\chi_k$, $k\in\RR$, defined in~\eqref{eq:def:chi_k} are linearly independent (in
  the algebraic sense, i.e., every finite subset is linearly independent).

  Moreover, for every $x_0\in\RR$ and $\eps>0$, the restrictions of these functions to the
  interval $\lt]x_0-\eps,x_0+\eps\rt[$ are linearly independent.
\end{lemma}
\begin{proof}
  We refer the reader to Appendix~\ref{apx:proofs:lem:chindependence} for the first statement and
  only prove the second one.

  Suppose that for some finite set $K\subset \RR$ and complex numbers $\alpha_k$, $k\in K$ we have
  $g(z) := \sum_{j=1}^m \alpha_j \chi_{k_j}(z) = 0$ for all $z\in\lt]x_0-\eps,x_0+\eps\rt[$.  Since~$g$
  is analytic and non-zero analytic functions can only vanish on a discrete set, we then must also
  have $g(z) = 0$ for all $z\in\CC$.  This means that the linear dependence on an interval implies
  linear dependence on the whole real line.  This proves the second statement, and the proof of
  Lemma~\ref{lem:chindependence} is completed.
\end{proof}

We can now give the definition of the quantity which will lower-bound the input redundancy for
linear input encoding.

\begin{definition}\label{def:FouRk}
  The \textit{Fourier rank} of a function $h\colon \RR \to \RR$ at a point~$x_0\in\RR$ is the
  infimum of the numbers~$r$ such that there exists an $\eps>0$, a set
  $K\subset\RR\setminus\{0\}$ of size~$2r$, and coefficients $\alpha_k\in\CC$, $k\in \{0\}\cup K$
  such that
  \begin{equation}\label{eq:fou-rep}
    h(x) = \sum_{k \in \{0\}\cup K} \alpha_k \,\chi_k(x) \quad\text{for all $x\in\lt]x_0-\eps,x_0+\eps\rt[$.}
  \end{equation}
\end{definition}

Note that the Fourier rank can be infinite, and if it is finite, then it is a nonnegative integer.
Indeed, from $h^*=h$ it follows that
$\sum_{k \in \{0\}\cup K} \alpha_k\chi_k = \sum_{k \in \{0\}\cup K} \alpha_k^* \chi_{-k}$, so that
by the linear independence of the $\chi$'s (Lemma~\ref{lem:chindependence}) we have
$\alpha_{-k} = \alpha_k^*$, which means that in a minimal representation of~$h$, the set~$K$ is
symmetric around~$0\in\RR$.

\textsc{Examples.}
\vspace{-1ex}%
\begin{itemize}
\item Constant functions have Fourier rank~0 at every point.
\item The trigonometric functions $x\mapsto\cos(\kappa x+\phi)$, with $\kappa\ne0$, have Fourier
rank~1 at every point.
\item Trigonometric polynomials of degree~$d$,
$x\mapsto\sum_{j=0}^d \alpha_j \cos^j(\kappa_j x+\phi_j)$, have Fourier rank~$d$ at every point,
if $\alpha_d\ne 0$, $\kappa_d \ne 0$.
\item The function $x\mapsto\abs{\sin(\pi x)}$ has Fourier rank~1 at every $x_0\in \RR\setminus\ZZ$ and
infinite Fourier rank at the points $x_0\in\ZZ$.
\item The function $x\mapsto x$ has infinite Fourier rank at every point.
\end{itemize}

\begin{theorem}\label{thm:spread-ge-FouRk}
  Let~$f$ be an expectation value function, i.e., as in
  Remark~\ref{rem:fourier-expansion}. Moreover, let $a, b\in\RR^n$, and
  $h\colon \RR\to\RR\colon x \mapsto f(x \cdot a + b)$.
  For every $x_0\in \RR$, the Fourier rank of~$h$ at~$x_0$ is less than or equal to the spread
  of~$a$.
\end{theorem}
\begin{proof}
  With the preparations above, this is now a piece of cake.  Let $x_0\in\RR$ and set $\eps:=1$.  With
  $K_a$ as defined in~\eqref{eq:def-spread}, for $x\in\lt]x_0-\eps,x_0+\eps\rt[$, we have
    \begin{align*}
      h(x) = f(x \cdot a + b)
      &= \sum_{w\in\ZZ_3^n} \hat f(w) e^{2\pi i w\bullet(x \cdot a + b)} &&\comment{Remark~\ref{rem:fourier-expansion}\ref{rem:fourier-expansion:C}}
      \\
      &= \sum_{w\in\ZZ_3^n} \hat f(w) e^{2\pi i w\bullet b} e^{2\pi i x\cdot w\bullet a}
      \\
      &= \sum_{k \in K_a} \Bigl(
        \sum_{\substack{w\in\ZZ_3^n,\\ w \bullet a = k}} \hat f(w) e^{2\pi i w\bullet b}
      \Bigr) \; e^{2\pi i x\cdot k}
      \\
      &= \sum_{k \in K_{a}} \alpha_k \, \chi_k(x),
    \end{align*}
    where we let
    \begin{equation*}
      \alpha_k := \sum_{\substack{w\in\ZZ_3^n,\\ w \bullet a = k}} \hat f(w) e^{2\pi i w\bullet b}
    \end{equation*}
    This shows that~$h$ has a representation as in~\eqref{eq:fou-rep} with
    $K := K_a\setminus\{0\}$.  It follows that the Fourier rank of~$h$ is bounded from above by
    $\abs{K_a}/2 = \spread{a}$.
    This completes the proof of Theorem~\ref{thm:spread-ge-FouRk}.
\end{proof}

The theorem allows us to give the concrete lower bounds for the input redundancy.

\begin{corollary}\label{cor:main-linear}
  Let~$h$ be a real-valued function defined in some neighborhood of a point $x_0\in\RR$.

  Suppose that in a neighborhood of $x_0$, $h$ is equal to an expectation value function with
  linear input encoding, i.e., there is an~$n$, a function~$f$ as in
  Remark~\ref{rem:fourier-expansion}, vectors $a,b\in\RR^n$, and an $\eps > 0$ such that
  $h(x) = f(x\cdot a+b)$ holds for all $x\in\lt]x_0-\eps,x_0+\eps\rt[$.

  The input redundancy, $n$, is greater than or equal to $\log_3(r+1)$, where~$r$ is the Fourier
  rank of~$h$ at~$x_0$.
\end{corollary}

\begin{center}
\fbox{%
  \begin{minipage}[h]{.75\linewidth}
    To represent a function~$h$ by a MiNKiF PQC with linear input encoding in a tiny neighborhood
    of a given point~$x_0$, the input redundancy must be at least the logarithm of the Fourier
    rank of~$h$ at~$x_0$.
  \end{minipage}%
}%
\end{center}

\begin{proof}[Proof of Corollary~\ref{cor:main-linear}]
  For every~$a\in\RR^n$, we have $\abs{K_a}\le 3^n$, by the definition of~$K_a$, and hence $\spread(a)\le (3^n-1)/2$.

  We allow that $a,b$ are chosen depending on~$h$ (see the Remark~\ref{rem:Var-Inp-Enc:lin} below).
  Theorem~\ref{thm:spread-ge-FouRk} gives us the inequality
  \begin{equation*}
    r \le \max_{a}\spread(a) \le (3^n-1)/2,
  \end{equation*}
  which implies $n \ge \log_3(2r+1) \ge \log_3(r+1)$, as claimed.  (We put the $+1$ to make the
  expression well-defined for $r=0$.)
  This concludes the proof of Corollary~\ref{cor:main-linear}.
\end{proof}

\begin{remark}\label{rem:Var-Inp-Enc:lin}
  If the entries of~$a$ are all equal up to sign, then we have $\spread(a)=n$.  It can be seen
  that if the entries of~$a$ are chosen uniformly at random in $[0,1]$, then
  $\spread(a)=(3^n-1)/2$.  Hence, it seems that some choices for~$a$ are better than others.
  Moreover, looking into the proof of Theorem~\ref{thm:spread-ge-FouRk} again, we see that the
  $\chi_k$, $k\in K_a$, must suffice to represent (or approximate) the target function, and that
  the entries of~$b$ play a role in which coefficients $\alpha_k$ can be chosen for a given~$a$.
  Hence, it is plausible that the choices of $a,b$ should depend on~$h$.
\end{remark}

\begin{remark}
  Our restriction to Hamiltonians with two eigenvalues leads to the definition of the spread
  in~\eqref{eq:def-spread}.  If the set of eigenvalue distances of the Hamiltonian encoding the
  input $\eta_j$ is $D_j \subset \RR$, then, for the definition of the spread, we must put this:
  \begin{equation*}
    K_a := \bigl\{ w\bullet a \bigm| w \in\prod_{j=1}^n D_j \bigr\}.
  \end{equation*}
  Theorem~\ref{thm:spread-ge-FouRk} and Corollary~\ref{cor:main-linear} remain valid, with
  essentially the same proofs, but with a higher base for the logarithm.
\end{remark}

\section{Arcsine input encoding}\label{sec:arcsin}
We now consider the original situation of the example
in~\cite{Mitarai-Negoro-Kitagawa-Fujii:q-circ-learn:2018}, where the activation function is
$\varphi = \arcsin$.  More precisely, for $a,b\in\RR^n$, we consider
\begin{equation*}
  \eta(x) := \arcsin( (a x + b)/(2\pi) ).
\end{equation*}
Abbreviating $s_j := a_j x + b_j$ and $c_j := \sqrt{1 - s_j^2}$ for $j=1,\dots,n$,
Remark~\ref{rem:fourier-expansion}\ref{rem:fourier-expansion:R}, gives us that the expectation
value functions with arcsine input encoding are of the form
\begin{equation}\label{eq:arcsin:expval-scmono}
  h(x)
  = f(\eta(x))
  =
  \sum_{\substack{S,C\subseteq[n]\\S\cap C = \emptyset}}
  \tilde f_{S,C} \prod_{j\in S} s_j \prod_{j\in C} c_j
  =
  \sum_{\substack{S,C\subseteq[n]\\S\cap C = \emptyset}}
  \tilde f_{S,C} \prod_{j\in S} (a_j x + b_j) \prod_{j\in C} \sqrt{1-(a_j x + b_j)^2},
\end{equation}
where we use the common shorthand $[n] := \{1,\dots,n\}$, and set
$\tilde f_{S,C} := \tilde f_{\tau(S,C)}$ with $\tau_j(S,C) = \sin$ if $j\in S$, $\tau_j(S,C)=\cos$
if $j\in C$, and $\tau_j(S,C)=\id$ otherwise.

Consider a formal expression of the form
\begin{equation}\label{eq:def:sc-mono}
  \mu^{(a,b)}_{S,C} := \prod_{j\in S} (a_j x + b_j) \prod_{j\in C} \sqrt{1-(a_j x + b_j)^2}
\end{equation}
where~$x$ is a variable (for arbitrary $a,b\in\RR^n$ and $S,C\subseteq[n]$ with
$S\cap C = \emptyset$).  We call it an \textit{sc-monomial} of degree $\abs{S}+\abs{C}$.
An sc-monomial can be evaluated at points $x\in\RR$ for which the expression under the square root
is not a negative real number, i.e., in the interval
\begin{equation}\label{eq:def:sc-mono-Int}
  I_\mu := \bigcap_{j\in C} \, \lt] \tfrac{-1-b_j}{a_j} , \tfrac{+1-b_j}{a_j} \rt[
\end{equation}
(which could be empty), and it defines an analytic function there.  Note, though, that it can
happen that an sc-monomial can be continued to an analytic function on a larger interval than
$I_\mu$.  The obvious example where that happens is this: For $j,j'\in C$ with $j\ne j'$ we have
$(a_j,b_j)=\pm(a_{j'},b_{j'})$.  In that case, the formal power series of the sc-monomial
simplifies, and omitting the interval $\lt] \frac{-1-b_j}{a_j} , \frac{+1-b_j}{a_j} \rt[$ (also
for~$j'$) from~\eqref{eq:def:sc-mono-Int} makes the intersection larger.

The following technical fact can be shown (cf.~\cite{GilVidal:PhD:2020}).

\begin{lemma}\label{lem:convrad-poly}
  Let $g = \sum_j \alpha_j \mu_j$ be a linear combination of sc-monomials with degrees at
  most~$d$, and suppose that $\bigcap_j I_{\mu_j} \ne \emptyset$.
  If an analytic continuation of~$g$ to a function $\tilde g\colon\RR\to\RR$ exists,
  then~$\tilde g$ is a polynomial of degree at most~$d$.
\end{lemma}

From this lemma, we obtain the following result.

\begin{corollary}\label{cor:arcsin-poly}
  Let~$h\colon\RR\to\CC$ be an analytic function, and $x_0\in\RR$.

  Suppose that in a neighborhood of $x_0$, $h$ is equal to an expectation value function with
  arcsine input encoding, i.e., there is an~$n$, a function~$f$ as in
  Remark~\ref{rem:fourier-expansion}, vectors $a,b\in\RR^n$, and an $\eps > 0$ such that
  \begin{enumerate}
  \item $-1\le x\cdot a_j+b_j\le+1$ for all $x\in\lt]x_0-\eps,x_0+\eps\rt[$, and
  \item $h(x) = f(\arcsin(x\cdot a+b))$ holds for all $x\in\lt]x_0-\eps,x_0+\eps\rt[$.
  \end{enumerate}
  Then
  $h$ is a polynomial, and
  the input redundancy, $n$, is greater than or equal to the degree of~$h$.
\end{corollary}

\begin{center}
\fbox{%
  \begin{minipage}[h]{.75\linewidth}
    To represent a polynomial~$h$ by a MiNKiF PQC with arcsine input encoding in a tiny neighborhood
    of a given point~$x_0$, the input redundancy must be at least the degree of~$h$.
  \end{minipage}%
}%
\end{center}

\begin{proof}[Proof of Corollary~\ref{cor:arcsin-poly}]
  Let us abbreviate
  $g\colon x\mapsto f(\arcsin(x\cdot a+b))\colon \lt]x_0-\eps,x_0+\eps\rt[\to\RR$.  From the
  discussion above, we know that~$g$ is a linear combination of sc-monomials.

  Both functions~$h$ and~$g$ are analytic, and they coincide on an interval.  Hence, $g$ has an
  analytic continuation, $h$, to the real line so that Lemma~\ref{lem:convrad-poly} is applicable,
  and states that~$h$ is a polynomial with degree at most~$n$.  This completes the proof of
  Corollary~\ref{cor:arcsin-poly}.
\end{proof}

As indicated in the introduction, in the special case which is considered in
\cite{Mitarai-Negoro-Kitagawa-Fujii:q-circ-learn:2018} --- where the input amplitudes are stored
(by rotations) in~$n$ distinct qubits before any other quantum operation is performed --- this can
be proved by looking directly at the effect of a Pauli transfer matrix on the mixed state vector
in the Pauli basis.  Our corollary shows that this effect persists no matter how the
arcsine-encoded inputs are spread over the quantum circuit.

The corollary allows us to lower bound the input redundancy for some functions.

\textsc{Examples.} %%
There is no PQC with arcsine input encoding that represents the function $x\mapsto \sin x$
(exactly) in a neighborhood any point.  Indeed, the same holds for any analytic function defined
on the real line which is not a polynomial: the exponential function, the sigmoid function, arcus
tangens, \dots

Unfortunately, from these impossibility results, no approximation error lower bounds can be
derived.  Indeed, in their paper~\cite{Mitarai-Negoro-Kitagawa-Fujii:q-circ-learn:2018} Mitarai et
al.\ point out that, due to the $\sqrt{\cdot}$ terms, the functions represented by the
expectation values can more easily represent a larger class of functions than polynomials.

To give lower bounds for the representation of functions which are not analytic on the whole real
line, we proceed as follows.
For fixed $n\ge 1$, $x_0\in\RR$ and $a,b\in\RR^n$, denote by $M^{n;a,b}_{x_0}$ the vector
space spanned by all functions of the form
$\lt]x_0-\eps,x_0+\eps\rt[ \to \RR \colon x \mapsto f(\arcsin(x\cdot a+b))$ for
an\footnote{Mathematically rigorously speaking, $M^{n;a,b}_{x_0}$ is the germ of functions
  at~$x_0$.}  $\eps>0$, where~$f$ ranges over all expectation value functions with input
redundancy~$n$, i.e., functions as in Remark~\ref{rem:fourier-expansion}, $a,b\in\RR^n$ satisfy
$-1<a_jx_0+b_j<+1$, and the $\arcsin$ is applied to each component of the vector.

\begin{proposition}\label{prop:vecdim:arcsin}
  The vector space $M^{n;a,b}_{x_0}$ has dimension at most~$3^n$, and is spanned by the
  sc-monomials~\eqref{eq:def:sc-mono} of degree~$n$.
\end{proposition}
\begin{proof}
  With $a,b$ fixed, there are at most $3^n$ sc-monomials~\eqref{eq:def:sc-mono} of degree~$n$, as
  $S,C\subseteq [n]$ and $S\cap C = \emptyset$ hold.  Hence, the statement about the dimension
  follows from the fact that the elements of $M^{n;a,b}_{x_0}$ are generated by sc-monomials.

  The fact that the sc-monomials generate the expectation value functions with arcsine
  input-encoding of redundancy~$n$ is just the statement of~\eqref{eq:arcsin:expval-scmono} above.
  This concludes the proof of Proposition~\ref{prop:vecdim:arcsin}.
\end{proof}

We can now proceed in analogy to the case of linear input encoding.  Let us define the
\textit{sc-rank} at~$x_0$ of a function~$h$ defined in a neighborhood of~$x_0$ as the infimum over
all~$r$ for which there exist sc-monomials $\mu_1,\dots,\mu_r$, an $\eps>0$, and real numbers
$\alpha_1,\dots,\alpha_r$ such that $x_0\in\bigcap_j I_{\mu_u}$, and
\begin{equation*}
  h(x) = \sum_{j=1}^r \alpha_j \mu_j(x)  \quad\text{for all $x\in\lt]x_0-\eps,x_0+\eps\rt[$.}
\end{equation*}

Proposition~\ref{prop:vecdim:arcsin} now directly implies the following result.

\begin{corollary}\label{cor:main-arcsine}
  Let~$h$ be a real-valued function defined in some neighborhood of a point $x_0\in\RR$.

  Suppose that in a neighborhood of $x_0$, $h$ is equal to an expectation value function with
  arcsine input encoding, i.e., there is an~$n$, a function~$f$ as in
  Remark~\ref{rem:fourier-expansion}, vectors $a,b\in\RR^n$, and an $\eps > 0$ such that
  $h(x) = f(\arcsin(x\cdot a+b)/(2\pi))$ holds for all $x\in\lt]x_0-\eps,x_0+\eps\rt[$.

  The input redundancy, $n$, is greater than or equal to $\log_3(r)$, where~$r$ is the sc-rank
  of~$h$ at~$x_0$.
\end{corollary}

\begin{center}
\fbox{%
  \begin{minipage}[h]{.75\linewidth}
    To represent a function~$h$ by a MiNKiF PQC with arcsine input encoding in a tiny neighborhood
    of a given point~$x_0$, the input redundancy must be at least the logarithm of the sc-rank
    of~$h$ at~$x_0$.
  \end{minipage}%
}%
\end{center}

We conclude the section with a note on the choice of the parameters $a,b$.

\begin{remark}\label{rem:Var-Inp-Enc:arcsin}
  It can be seen~\cite{GilVidal:PhD:2020} that the dimension of the space $M^n_{x_0}$ is~$3^n$, if
  $a_j,b_j$ $j=1,\dots,n$ are chosen in general position, but only $O(n)$ if $a$ is a constant
  multiple of the all-ones vector.  Moreover, as indicated in
  Proposition~\ref{prop:vecdim:arcsin}, the basis elements which span the space depend on $a,b$,
  and hence the space $M^{n;a,b}_{x_0}$ will in general be different for different choices of
  $a,b$.
  Again, we find that it is plausible that the choices of $a,b$ should depend on the target
  function.
\end{remark}

\section{Conclusions and outlook}\label{sec:outlook}
To the best of our knowledge, our results give the first rigorous theoretical quantitative
justification of a routine decision for the design of parameterized quantum circuit architectures:
Input redundancy \textsl{must} be present if good approximations of functions are the goal.

Both activation functions we have considered give clear evidence that input redundancy is
necessary, and grows at least logarithmically with the ``complexity'' of the function: The
complexity of a function~$f$ with respect to a family $\mathcal B$ of ``basis functions'' is the number of
functions from the family which are needed to obtain~$f$ as a linear combination.  In our results,
the function family~$\mathcal B$ depends on the activation function.  In the case of linear input
encoding (activation function ``identity''), the basis functions are trigonometric functions
$t\mapsto e^{2\pi i kt}$, whereas for the $\arcsin$ activation function, we obtain the basis
monomials~\eqref{eq:def:sc-mono} already used, in a weaker form,
in~\cite{Mitarai-Negoro-Kitagawa-Fujii:q-circ-learn:2018}.

From Remarks \ref{rem:Var-Inp-Enc:lin} and~\ref{rem:Var-Inp-Enc:arcsin} we see that that the
weights $a,b$, i.e., the coefficients in the affine transformation links in Fig.~\ref{fig:qnn},
should have to be variable in order to ensure a reasonable amount of expressiveness in the
function represented by the quantum circuit.  We use the term \textit{variational input encoding}
to refer to the concept of training the parameters involved in the encoding with other model
parameters.  A recent set of limited experiments~\cite{AndrewLei:MSc:2020} indicate that
variational input encoding improves the accuracy of Quantum Neural Networks in classification
tasks.

%\subsection{Outlook}
\par\noindent
While we emphasize the point that this paper demonstrates a concept --- lower bounds for input
redundancy can be proven --- there are a few obvious avenues to improve our results.

Most importantly, our proofs rely on exactly representing a target function.  This is an
unrealistic scenario.  The most pressing task is thus to give lower bounds on the input
redundancy when an approximation of the target function with a desired accuracy $\eps > 0$ in a
suitable norm is sufficient.

Secondly, we thank an anonymous reviewer for pointing out to us that lower bounds for many more
activation functions could be proved.

Finally, Remark~\ref{rem:fourier-expansion} mentions that for every multi-linear trigonometric
polynomial~$f$, there is a PQC whose expectation value function is precisely~$f$.  It would be
interesting to lower-bound a suitable quantum-complexity measure of the PQCs representing a
function, e.g., circuit depth.  While comparisons of the quantum vs classical complexity of
estimating expectation values have attracted some attention~\cite{bravyi2019classical}, to our
knowledge, the same question in the ``parameterized setting'' has not been considered.

\subsection*{Conflict of Interest Statement}
Author Dirk Oliver Theis was employed by the company Ketita Labs O\"U.  The remaining author
declares that the research was conducted in the absence of any commercial or financial
relationships that could be construed as a potential conflict of interest.

\subsection*{Author Contributions}
DOT contributed the realization that lower bounds could be obtained, and sketches of the proofs.
FJGV contributed the details of the proofs, and the literature overview.

\subsection*{Funding}
This research was supported by the Estonian Research Council, ETAG (\textit{Eesti
  Teadusagentuur}), through PUT Exploratory Grant \#620.  DOT is partly supported by the Estonian
Centre of Excellence in IT (EXCITE), funded by the European Regional Development Fund.

%%%%%%%%%%%%%%%%%%%%%%%%%%%%%%%%%%%%%%%%%%%%%%%%%%%%%%%%%%%%%%%%%%%%%%%%%%%%%%%%%%%%%%%%%%%%%%%%%%%%%%%%%%%%%%%%
\bibliographystyle{plain}
\bibliography{Javier}
\appendix
%\section{Deferred proofs}\label{apx:proofs}
\newsavebox{\sldfweifj}
\savebox{\sldfweifj}{\ref{lem:chindependence}}
\section{Proof of Lemma~\usebox{\sldfweifj}}\label{apx:proofs:lem:chindependence}
We have to prove that the functions $\chi_k$, $k\in \RR$, defined in~\eqref{eq:def:chi_k} are
linearly independent (in the algebraic sense, i.e., considering finite subsets of the functions at
a time).  There are several ways of proving this well-known fact; we give the proof that probably
makes most sense to a physics readership: The Fourier transform (in the sense of tempered
distributions) of the function $\chi_k$ is $\delta(k-*)$, the Dirac distribution centered on~$k$.
These generalized functions are clearly linearly independent for different values of~$k$.

\end{document}